\documentclass[runningheads]{llncs}
\usepackage[T1]{fontenc}
\usepackage{graphicx}
\usepackage{color}
\usepackage[dvipsnames]{xcolor}

\usepackage{amsmath,amssymb,amsfonts,amsthm}
\usepackage{hyperref}
\usepackage{multicol}
\usepackage{array}    
\usepackage{tabularx} 
\usepackage[linesnumbered,ruled,vlined,algo2e]{algorithm2e}

\usepackage{cleveref}
\usepackage{xcolor}
\usepackage{enumitem}
\usepackage{thmtools}
\usepackage{multirow}

\usepackage{comment}
\usepackage{tikz}
\usetikzlibrary{decorations.pathreplacing}
\usetikzlibrary{fit, backgrounds}
\usetikzlibrary{graphs, arrows.meta, decorations.pathmorphing, calc}
\usepackage{thm-restate}

\usepackage{sidecap}
\theoremstyle{plain}

\begin{document}
\title{Monotone Decontamination of Arbitrary Dynamic Graphs with Mobile Agents\thanks{Eligible for best student paper award.}}

\titlerunning{Monotone Decontamination of Arbitrary Dynamic Graphs}

\author{Rajashree Bar\inst{1} \and Daibik Barik\inst{2} \and Adri Bhattacharya\inst{1} \and Partha Sarathi Mandal\inst{1}}

\authorrunning{Bar et al.}

\institute{Indian Institute of Technology Guwahati, Guwahati, India\\ \and Indian Statistical Institute, Bangalore, India}

\maketitle              
\begin{abstract}

Network decontamination is a well-known problem, in which the aim of the mobile agents should be to decontaminate the network (i.e., both nodes and edges). This problem comes with an added constraint, i.e., of \emph{monotonicity}, in which whenever a node or an edge is decontaminated, it must not get recontaminated. Hence, the name comes \emph{monotone decontamination}. This problem has been relatively explored in static graphs, but nothing is known yet in dynamic graphs. We, in this paper, study the \emph{monotone decontamination} problem in arbitrary dynamic graphs. We designed two models of dynamicity, based on the time within which a disappeared edge must reappear. In each of these two models, we proposed lower bounds as well as upper bounds on the number of agents, required to fully decontaminate the underlying dynamic graph, monotonically. Our results also highlight the difficulties faced due to the sudden disappearance or reappearance of edges. Our aim in this paper has been to primarily optimize the number of agents required to solve monotone decontamination in these dynamic networks.

\end{abstract}

\section{Introduction}
We consider a connected network where nodes and edges are infected by various attacks, including but not limited to, viruses, spam, and malware. In real life scenarios, the concern is to keep the network safe from such harmful contamination. In particular, we work on dynamic graphs where the edges can disappear and reappear at any time, provided the underlying graph remains connected. This phenomenon of connectivity is termed as \textit{1-interval connectivity}. In literature, the decontamination has been achieved by using a team of mobile agents, called searchers. Each node of a graph can be in one of three states: contaminated or decontaminated. There are certain variations of contamination which has been studied in literature. A brief survey of which can be found in \cite{Nisse2019}. Some of these variations are: \emph{Edge-search} \cite{parsons1978search}, \emph{Node-search} \cite{kirousis1985interval} and \emph{Mixed-search} \cite{bienstock1991monotonicity}. The \emph{edge-search} phenomenon is a decontamination strategy where an agent traverses through an edge and then decontaminates. The \emph{node-search} phenomenon is slightly different from the edge-search model. Here, an edge is clean when both of its endpoints are occupied by some agents. And lastly, the \emph{mixed-search} is an edge-search strategy with the difference that an edge becomes clean when either an agent traverses through that edge or its endpoints are guarded by at least one agent. Network decontamination is well studied in static graph settings such as hypercubes \cite{flocchini2008decontamination}, tori, and chordal rings \cite{FLOCCHINI2007}. As per our knowledge, we are the first ones to explore this decontamination problem on dynamic graphs. We consider two models of dynamicity, in the first one, the adversary can disappear or reappear certain edges (adhering to 1-interval connectivity), but those edges that are disappeared can only remain disappeared for a finite time. We term this model as \textit{Finite time edge appearance} model or FTEA. In the second model, there is no bound on disappearing time, i.e., in other words an edge can remain disappeared for an infinitely long time. We term this model as \textit{Indefinite edge disappearance} model or IDED. FTEA is a special case of IDED. In this paper, we design algorithms and lower bounds on the number of agents for both these models of dynamicity.


\noindent \textbf{Related Works:} Network decontamination, often referred to as graph searching, is a problem that has been studied for decades, beginning with Parsons et al. in the year 1976~\cite{parson1976pursuit}. The basic idea is simple: a group of mobile agents (searchers) tries to clean up a contaminated network, using as few agents as possible. Without any knowledge of the underlying graph, identifying the minimum number of searchers needed is actually NP-complete for general graphs~\cite{megiddo1988complexity}. Luccio et al. \cite{luccio2007network} introduced a different model of decontamination, termed as $m$-immunity. It is defined to be as follows: any decontaminated node with no agent on it is further recontaminated, only if at least $m$ of its neighbors are contaminated. They solved these problem for trees, tori, mesh networks and for graphs with maximum degree 3. Note that in this paper, the value of $m$ is limited to $m=1$ or to the strict majority of the neighbor of the nodes. Further, Flocchini et al. \cite{Flocchini2016mImmunity} extended this problem for any integer value $m\ge 1$. They solved this problem for the underlying graphs to be mesh networks, trees and hypercubes. All these works so far studied in literature, assume the network to be static (i.e., does not change over time). Accordingly, monotone strategies (i.e., once something is cleaned, it stays clean) is designed to decontaminate the underlying graph.

In contrast to static networks, dynamic graphs allow edges to disappear and reappear over time, subject to a connectivity constraint (referred to as 1-interval connectivity, meaning the graph remains connected at every round). Trivially, in comparison to static networks, solving these problems in dynamic networks adds more complexity. Recently, there has been a lot of interest in studying the fundamental problems, such as exploration \cite{das2019graph}, gathering \cite{di2020gathering}, dispersion \cite{kshemkalyani2020efficient}, black hole search \cite{bhattacharya2024black},\cite{Kaur2025}, etc., in dynamic graphs. Accordingly, in this paper, we explored this decontamination problem for dynamic networks. As per our knowledge, this is the first paper to establish bounds on monotone decontamination in dynamic networks.


\noindent\textbf{Our Contribution:} 
In this paper, our contributions are as follows.\\
Let $\mathcal{G}$ be any dynamic graph with $n$ nodes, and let $\mu(G)=k$ be the cyclomatic number of the initial static graph $G$ of $\mathcal{G}$. Based on the two models of dynamicity, we obtain the following results.

\noindent\textbf{Finite time edge reappearance Model (FTEA):} In this model, each edge that disappears must reappear within $T$ time. Based on the relation between $k$ and $n$, we obtain the following results.

\noindent\textbf{1.} For $k<n$, we show that if $d$ is the diameter of $G$ then at least $d$ agents are required to solve monotone network decontamination. Further, we proposed an algorithm that requires $d+k$ agents to solve monotone decontamination on $\mathcal{G}$.

\noindent\textbf{2.} For $k\ge n$, we show that at least $n-1$ agents are required to solve monotone decontamination. Accordingly, we proposed an algorithm that solves monotone decontamination in this model of dynamicity with $n$ agents.

\noindent\textbf{Indefinite Edge Disappearance Model (IDED):} 
In this model, edges may disappear for an unbounded period of time. This adds more complexity in comparison to the earlier FTEA model, since FTEA is a special case of the IDED model. Based on the relation between $k$ and $n$, we have the following results.

\noindent\textbf{1.} For $k<n$, we have shown that at least $k+1$ agents are required to monotonically decontaminate a graph in this model. But, since FTEA is a special case of IDED, the bound for FTEA also holds for IDED. This means, if $d>k+1$, then as per the result obtained in the FTEA model, we can say that at least $d$ agents are required to solve monotone decontamination in this model. Hence, this gives our lower bound to be $\max \{d,k+1\}$.

\noindent\textbf{2.} For $k\ge n$, again as per the result obtained in the FTEA model, we propose that at least $n-1$ agents are required to solve network decontamination under IDED model as well.

Lastly, irrespective of $k<n$ or $k\ge n$, we propose a monotone network decontamination algorithm in the IDED model that requires $d+2k$ agents. Summary of the results show in Table \ref{tab:summary}


\begin{table}[h]
\centering
\small
\renewcommand{\arraystretch}{1.4} 
\setlength{\tabcolsep}{7pt}                    
\begin{tabularx}{\linewidth}{|c|c|X|X|}
\hline
\textbf{Model} & \textbf{Case} & \textbf{Upper Bound} & \textbf{Lower Bound} \\
\hline
\textbf{FTEA} 
  & $k < n$ 
  & $d + k$ (Theorem \ref{d+k result}) 
  & $d$ (Lemma \ref{lem:tree_lb}) \\

\cline{2-4}
  & $k \ge n$ 
  & $n$ (Theorem \ref{n result}) 
  & $n - 1$ (Theorem \ref{n-1 lb}) \\
\hline
\textbf{IDED} 
  & $k < n$ 
  & \multirow{2}{=}{\centering $d + 2k$ (Theorem \ref{d+2k result})} 
  & $\max\{d,\,k+1\}$ (Lemmas \ref{lem:tree_lb},\ref{k+1 wheel}) \\
\cline{2-2}\cline{4-4}
  & $k \ge n$ 
  &  
  & $n-1$ (Theorem \ref{n-1 lb}) \\
\hline
\end{tabularx}

\caption{Summary of the results. 
The underlying graph is $G=(V,E)$ with $|V|=n$, diameter $d$, 
and cyclomatic number $k$.}
\label{tab:summary}
\end{table}


\section{Model and Preliminaries}
\label{sec:model} 

\noindent\textbf{Dynamic Graph Model:} The dynamic graph is modeled as a {\it time-varying graph} termed as TVG, where it is denoted by $\mathcal{G}=(V,E,\mathbb{T},\rho)$, where $V$ indicates the set of vertices, $E$ indicates the set of edges, $\mathbb{T}$ is said to be the {\it temporal domain}, which is symbolized to be $\mathbb{Z}^+$, as we consider discrete time steps. Further we define $\rho: E\times \mathbb{T}\rightarrow \{0,1\}$ as the {\it presence} function, which indicates whether an edge is present or absent at any given time. The graph $G = (V, E)$ is the initial underlying static graph of the \emph{time varying graph} $\mathcal{G}$. This underlying graph $\mathcal{G}$, with $|V|$ many vertices and $|E|$ many edges, is stated to be the footprint of $\mathcal{G}$. The degree of a node $v\in V$ is denoted by $\delta_v$, where $\Delta$ denotes the maximum degree in $G$. The footprint graph is an undirected anonymous graph (i.e., the nodes in $G$ have no IDs). Each edge with respect to a node $v\in V$ is uniquely labeled by a port numbers which ranges from $\{0,1,\dots,\delta_{v}-1\}$. The adversary has the ability to disappear or reappear any edge from $G$ at any round. This disappearance or reappearance of edge(s) is done by the adversary, keeping in mind that the underlying graph remains connected. This connectivity property is termed as {\it 1-interval connectivity}.

In this paper, we consider two dynamic models based on the time the adversary can make an edge disappear. The first model is defined as {\it finite time edge reappearance model} (FTEA), where there exists a constant $T$ ($>0$) for which at most, the adversary can make any edge disappear, at time $T+1$, it has to make that edge reappear. The Second model is defined as {\it indefinite edge disappearance} time (IDED), where there is no existence of such a $T$, i.e., the adversary can make an edge disappear for any amount of time (i.e., even not finite time). Note that, in both case, the adversary must maintain 1-interval connectivity property in the underlying graph. 

\noindent\textbf{Agent Model:} The agents are initially co-located at a node $r\in V$, chosen by the adversary. This node $r$ is also termed as \emph{Home}. They work synchronously. So, time is calculated in terms of rounds. The agents have distinct IDs, and they do not have any knowledge about any graph parameters, such as $|V|$, $|E|$, $\Delta$, etc. When an agent visits any node $v\in V$, it can know the degree of that node in the footprint graph $G$. The agents can see the presence or absence of the edges at that round. For each adjacent edge present, the agents can see whether it is contaminated or not. The agents communicate via the {\it face-to-face} model of communication, where they can share information with another agent when they are at the same node in the same round.  

\noindent\textbf{Decontamination Model:} Initially, we define all the edges and vertices to be \emph{contaminated}. Whenever an agent visits a node, or traverses through an edge, that specific node or edge becomes \emph{decontaminated}. 
\emph{Recontamination} can occur if there exists a path, unguarded by agents, from a contaminated component to the decontaminated component. For example, in a cycle graph $C_3$, an agent starts from a node $v_1$, making it initially decontaminated. It goes to $v_2$, decontaminating the edge $(v_1, v_2)$. But since the vertex $v_3$ is still contaminated and it's directly joined with $v_1$ and hence with the edge $(v_1, v_2)$, it recontaminates both of them, but not the node $v_2$ because the agent is still present on it.

Here, we define the required definitions and problem definition.

\begin{definition}[Contamination Degree]
The set $\{C_t(v)\}$ is the collection of adjacent edges of the node $v$, which are contaminated at time $t$. We define it as $\{C_t(v)\} = \{(u, v) \in E_t \mid (u, v) \text{ is contaminated at time } t\}$. Moreover, $C_t(v)$ denotes the contamination degree of $v$ at time $t$. 
\end{definition}


\begin{definition}[Cyclomatic Number]
The cyclomatic number of a connected undirected graph $ G = (V, E) $ is given by
\[\mu(G) = |E| - |V| + 1.\]
Equivalently, $ \mu(G) $ is the number of independent cycles in $G$, and equals the number of edges that must be removed to obtain a spanning tree. 
\end{definition}

\begin{definition}[Feedback Edges\cite{kudelic2022feedback}]
Let $G=(V,E)$ be a static connected graph and $T=(V,E')\subseteq G$ be a fixed spanning tree of $G$.  The \emph{feedback edge set} is $\Tilde{E} = E\setminus E'$, and $ |\Tilde{E}|=k$ equals the cyclomatic number $\mu(G)$.
\end{definition}

\begin{definition}[Separator Vertex]
A vertex $ v \in V$ is called a \emph{separator vertex}, if say at round $t$, it is adjacent to a contaminated edge. The set of all separator vertices is denoted by $\Sigma$. Moreover, the agents present at these separator vertices can be called separators or separator agent. 

\end{definition}

\begin{definition}[Problem Definition]
Given an anonymous, port-labeled dynamic graph $\mathcal{G}$, an initial node $Home$, determine the minimum number of agents required to solve the following protocols:
\begin{itemize}
    \item Decontaminates all vertices and all existing edges of the underlying dynamic network.
    \item Ensures monotonicity by preventing any recontamination.
\end{itemize}
\end{definition}

\begin{remark}
    It may be noted that, in all our algorithms we assume the initial deployment of agents to be sufficient in number. This sufficient number depends on the underlying graphs diameter (i.e., $d$), total number of nodes (i.e., $n$) and the cyclomatic number (i.e., $k$). But, this knowledge is only confined to the initial deployment on the number of agents, but nowhere in our algorithms do the agents require the help of any of these global graph parameters.
\end{remark}


\section{Lower Bound}

In this section we propose the lower bound results obtained. We divide this section in to two parts, based on the dynamicity model.

\subsection{Lower Bound for Finite Time Edge Reappearance Model}
This argument gives the lower bound for the case of the finite-time edge reappearance model.

\begin{theorem}[Lower Bound]\label{n-1 lb}
There exists a graph $G$ with $n$ nodes, and maximum degree $\frac{n}{2}$, on which no deterministic algorithm can solve network decontamination monotonically with $n-2$ initially co-located agents, in the finite edge reappearance model of dynamicity, where $n>4$.
\end{theorem}

\begin{figure}[ht]
  \centering
  \begin{tikzpicture}[every node/.style={circle,draw,minimum size=6mm}, x=1.5cm]
    \foreach \i in {1,...,4} {
      \node (A\i) at (0,-\i) {$a_{\i}$};
    }
    \foreach \j in {1,...,4} {
      \node (B\j) at (3,-\j) {$b_{\j}$};
    }
    \foreach \i in {1,...,4} {
      \foreach \j in {1,...,4} {
        \draw (A\i) -- (B\j);
      }
    }
  \end{tikzpicture}
  \caption{The complete bipartite graph $K_{4,4}$ with partitions $\{a_1,\dots,a_4\}$ and $\{b_1,\dots,b_4\}$.}
  \label{fig:K44}
\end{figure}
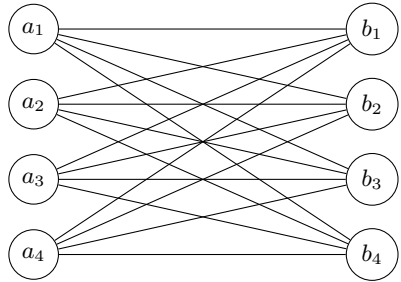
\begin{proof}
Let us consider our underlying graph to be $K_{n/2,n/2}$, where $A$ and $B$ are two disjoint partitions, each of size $\frac{n}{2}$, where $n>4$. Let the group of $n-2$ mobile agents, start from any node, we term that node as $Home$. Without loss of generality (WLOG), let us assume that our $Home$ node is $a_1\in A$ (refer to Fig. \ref{fig:K44}). A node is said to be \emph{clean}, when all its adjacent edges have been decontaminated by the agents. Now, as per the characteristics of $K_{n/2,n/2}$, the vertex $a_1$ has degree $\frac{n}{2}$. Let $\mathcal{A}$ be any algorithm that claims to solve the network decontamination problem monotonically.

Let $r$ ($>0$) be the round at which, during the execution of $\mathcal{A}$, an agent first visits the $\frac{n}{2}$-th adjacent node of $Home$. Let us call this node $b_t$ (where $t\in\{1,2,\dots,\frac{n}{2}\}$). Now, for $\mathcal{A}$, there can be two strategies. First, within round $r$ there exists at least one node (say, $b_{t1}$) in $B$ that is \textit{fully decontaminated}, i.e., in other words $b_{t1}$ along with all its adjacent edges are decontaminated, implying that its neighboring nodes are also decontaminated. Second, within round $r$, no such node of the form $b_{t1}$ exists. For each of these two strategies, we infer a contradiction in the following manner.

\noindent{\it Strategy-1:} Let, without loss of generality, $b_{t1}$ be the first node to be fully decontaminated, and let this happen at round $r_0$ ($<r$). This means all neighbour nodes of $b_{t1}$, i.e., each node in $A$ contains at least one agent. If any one of these agents holding a node in $A$ (say, that node is $a_{t1}$), decides to leave $a_{t1}$, then $a_{t1}$ gets recontaminated and monotonicity fails. The reason being, there exist at least two adjacent edges of $a_{t1}$ which are contaminated. As otherwise there exists at least one node in $A$, say $b_{t2}$ which is fully decontaminated before round $r$. It is because at round $r_0$, $\frac{n}{2}$ agents are holding a node each in $A$. Remaining, $\frac{n}{2}-2$ agents cannot fully decontaminate any node in $B$, as those nodes must have exactly $\frac{n}{2}-1$ contaminated edge. This concludes that none of the agents in $A$ can move after $b_{t1}$ is fully decontaminated. This also contradicts that $\mathcal{A}$ solves the problem, as no other node in $A$ can be fully recontaminated with the remaining vertices after round $r_0$. 

\noindent{\it Strategy-2:} After round $r$, all nodes in $B$ are occupied by at least one agent. So, before $r$, no node in $A$ can be fully decontaminated. So, each visited node in $A$ till round $r$ must contain at least one agent. Let there be $x$ many such nodes. Observe that at round $r$, one of these $x$ many nodes must be fully decontaminated (since each node in $A$ contains at least one vertex). So, at this moment, at least $x-1$ agents remain occupying the remaining $x-1$ yet to be fully decontaminated nodes in $A$. Now, observe that any node in $B$ must have at least $\frac{n}{2}-x$ many contaminated edges at round $r$. But since $\mathcal{A}$ is said to achieve monotonicity, so the remaining agents, that are free to move (or not occupying any contaminated node) is $\frac{n}{2}-x-1$ ($=n-2-(\frac{n}{2}+x-1)$), and these remaining agents cannot fully decontaminate any more node in $B$, which leads to a contradiction.

This shows that $\mathcal{A}$ fails to decontaminate $K_{n/2,n/2}$ with $n-2$ co-located agents, monotonically. \qed \end{proof}



\subsection{Lower Bound for Indefinite Edge Disappearance Model}
\label{sec:wheel-lower}

In this section, we show that there exists a graph, with cyclomatic number $\mu(G)=k$, on which any monotone decontamination protocol in the indefinite edge disappearance model with $k$ agents fails.

\begin{theorem}\label{k+1 wheel}
There exists a graph $G$ with $n$ (with $n>4$) nodes and maximum degree $n-1$, with $\mu(G)=n-1$, on which no deterministic algorithm can solve network decontamination monotonically with $\mu(G)$ many initially co-located agents, in the infinite edge disappearance model.
\end{theorem}

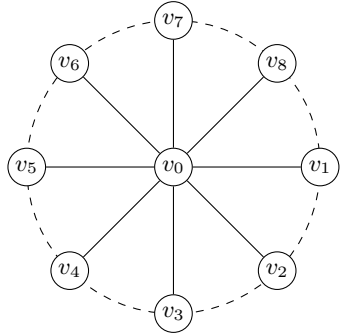
\begin{figure}[ht]
\centering
\begin{tikzpicture}[scale=1.3,
every node/.style={circle,draw,fill=white,minimum size=5mm,inner sep=1pt}]

    \def\R{1.5}

    \draw[dashed, thin] (0,0) circle (\R cm);

    \foreach \i in {1,...,8} {
        \node (N\i) at ({-45*(\i-1)}:\R cm) {$v_{\i}$};
    }

    \node (N0) at (0,0) {$v_{0}$};

    \foreach \i in {1,...,8} {
        \draw (N0) -- (N\i);
    }
\end{tikzpicture}
\caption{The wheel graph $G$ on 9 vertices with cyclomatic number $\mu(G)=8$, where the dotted edges imply the disappeared edges.}
  \label{fig:wheel-lower}
\end{figure}

\begin{proof}
In contradiction, let $\mathcal{A}$ be any deterministic monotone decontamination strategy using $n-1$ agents. Let us consider $G$ to be a wheel graph with the vertex set $V=\{v_0,v_1,\dots,v_{n-1}\}$, where $v_0$ is the centre vertex and the remaining are outer vertices (for reference see Fig. \ref{fig:wheel-lower}). Suppose that the adversary selects $v_0$ to be the initial node for these $n-1$ agents. In addition, suppose the adversary removes each edge of the form $(v_i,v_{i+1})$ ($i\in\{1,\cdots,n-2\}$) and $(v_{n-1},v_1)$. In this situation, to clean any vertex $v_i$ ($i\in\{1,\dots,n-1\}$), an agent must visit it through $v_0$. Let $r_0$ be the round, when any agent from $v_0$ first visits any node of the form $v_i$, say that the first node visited is $v_t$. Now, we claim that from round $r_0$ onward, there must exist at least one agent guarding the node $v_t$. As otherwise, the adversary can reappear the contaminated edges $(v_t,v_{t-1})$ or $(v_{t+1},v_t)$ and hence monotonicity fails. Since, this argument is true in general for each $v_i$. Hence, this proves that eventually $n-1$ agents get stuck at each of these vertices. Say that the round is $r'_0$ ($>r_0$), when each of the $n-1$ nodes is guarded by one agent. From $r'_0+1$ round onwards, if the adversary reappears on each edge of the form $(v_i,v_{i+1})$ ($i\in\{1,\cdots,n-2\}$) and $(v_{n-1},v_1)$, then none of these agents can move from their respective vertices. It is because the degree of contaminated edges from each $v_i$ is exactly 2, whereas there is only a single agent guarding $v_i$. This means, $G$ cannot be fully decontaminated using the strategy $\mathcal{A}$. Hence, this leads to a contradiction.\qed \end{proof}

\begin{corollary}
In the worst case, over all dynamic graphs with cyclomatic number $\mu(G)=k$, at least $k+1$ agents are necessary for monotone decontamination under adversarial edge dynamics.
\end{corollary}

Next, we discuss a lower bound result, that shows that on graphs with diameter $d$, at least $d$ agents are required to decontaminate those graphs under the \textit{finite time edge appearance} model.

\begin{lemma}\label{lem:tree_lb}
There exists a graph $G$ consisting of $n$ nodes with diameter $d$, on which there does not exist any deterministic monotone decontamination strategy that works with $d-1$ co-located agents, irrespective of the finite edge reappearance or indefinite edge disappearance model.
\end{lemma}

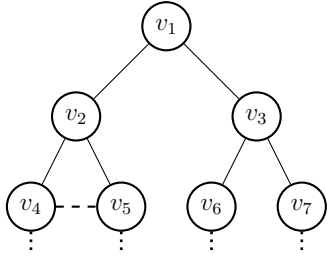
\begin{figure}
        \centering
        \begin{tikzpicture}[
    scale = 0.8, transform shape,
    every node/.style={circle, draw, thick, inner sep=2pt, minimum size=8mm, font=\large},
    level distance=1.5cm,
    level 1/.style={sibling distance=3cm},
    level 2/.style={sibling distance=1.5cm}
]

\node (v1) {$v_1$}           
    child {node (v2) {$v_2$}   
        child {node (v4) {$v_4$}} 
        child {node (v5) {$v_5$}} 
    }
    child {node (v3) {$v_3$} 
        child {node (v6) {$v_6$}}
        child {node (v7) {$v_7$}} 
    };

\draw[dashed, thick] (v4) -- (v5);

\draw[dotted, line width=1pt] (v4) -- ++(0,-0.8);
\draw[dotted, line width=1pt] (v5) -- ++(0,-0.8);
\draw[dotted, line width=1pt] (v6) -- ++(0,-0.8);
\draw[dotted, line width=1pt] (v7) -- ++(0,-0.8);

\end{tikzpicture}
\label{fig:tree-joined}
\caption{This shows a graph $G$ where $d$ is the diameter and a possible $k=1$. The dotted edge between $v_4$ and $v_5$ denotes the possibility for $k=1$.}
    \end{figure}

\begin{proof}
We prove this claim by considering the finite edge reappearance model, but it can be easily modified to the indefinite edge disappearance model. Let $v_1\in V$ be an initial vertex where $d-1$ agents are placed. We assume that, initially, $G$ is contaminated, i.e., all vertices and edges of $G$ are contaminated. Let $\mathcal{A}$ be such a strategy that claims to solve network decontamination with $d-1$ agents, starting from $v_1$. Note that, in the first round, all the agents can't leave $v_1$ to go, say $v_2$, because the contaminated edge $(v_3,v_1)$ recontaminates the unguarded node $v_1$. This prevents the \emph{monotonicity} of our process.

So, the only possibility that remains is for the agents to split up, such that one of the 2 agent groups (say $\{A_1\}$) visits $v_2$ and for the other agent group (say $\{A_2\}$) to wait or move forward to $v_3$. Following the recursion again, the agent group $\{A_1\}$ has to split into smaller groups to move forward among $v_4$ and $v_5$. Similarly, whenever $\{A_2\}$ moves forward along $(v_1,v_2)$, they have to split off to mover forward among $v_6$ and $v_7$ and this continues recursively.

As the size of the groups continue decreasing, and since $d<n/2$, there would be a round $r$ at which every agent, group is of size 1, visits a new parent node. From this point onwards, none of these agents, can move further. It is because, the degree of contaminated edges at parent nodes is 2 each, but a single agent is guarding them. Hence, $\mathcal{A}$ fails, which leads to a contradiction.\qed\end{proof}  

\begin{remark}
    Here in Lemma \ref{lem:tree_lb} we took the case of $k=0$. In fact, we can increase $k$ by connecting children up to $d-1$ bounded by the result in Lemma \ref{sec:wheel-lower}. We can still prove Lemma \ref{lem:tree_lb} with some modifications.
\end{remark}

\section{Algorithm for Network Decontamination}\label{sec:algo_decon}

In this section, we propose network decontamination algorithms for both models of dynamicity. We start with the \emph{finite edge reappearance model}.

\subsection{Finite Edge Reappearance Model}
\label{sec: Algorithm FiniteEdgeRe}

Under this section, we propose two monotone decontamination algorithms, first for the case $k\ge n$ and next for the case $k<n$.
\subsubsection{Case ($k\ge n$):} The following algorithm (termed as \textsc{Uni-Decontamination}) discusses a universal monotone decontamination strategy for the case $k\ge n$, where $\mu(G)=k$, $|V|=n$ and $G=(V,E)$ is the underlying static graph of the dynamic graph $\mathcal{G}$.

\noindent{\textbf{High-Level Idea of \textsc{Uni-Decontamination}:}} We present a deterministic strategy satisfying monotonic decontamination of a dynamic graph $\mathcal{G}$ using $n$ agents. The agents are initially co-located at a node, termed as $Home$. The algorithm proceeds in two main phases: the first phase is called \emph{dispersal} and the second phase is called \emph{cleaning}. In the \emph{dispersal} phase, the agents starting from $Home$ perform \textit{breadth first search} (or BFS). The moment a group of agent arrive at an empty node (i.e., not already occupied by any other agent), the lowest ID agent settles, and the rest continue to perform BFS. These stationary agents serve as guards, ensuring that once a vertex is cleaned, it cannot be recontaminated due to changes in the dynamic edge set. During this strategy, whenever an agent finds that the edge towards its destination node has disappeared, it waits for that edge to reappear. Now, once the $Home$ node is \emph{fully decontaminated}, i.e., $Home$ and its adjacent edges are decontaminated, then the agent at $Home$, say $a_1$, acts as a \emph{cleaner} agent and starts the second phase of the algorithm.

In the cleaning phase, $a_1$ begins to explore the network to identify and clean every remaining contaminated edge. The moment it encounters a disappeared contaminated edge, it waits for that edge to reappear, i.e., for at most $T$ times (refer to dynamic graph model in Section \ref{sec:model}). So, eventually every edge and vertex of $\mathcal{G}$ gets decontaminated, and the algorithm also maintains monotonicity.

\begin{algorithm2e}[h]\small
\caption{\textsc{Uni-Decontamination}($k,Home$)}
\label{algo:uni_decon}
\KwIn{$Home$ node, $n$ agents}

Initially $n$ agents are at $Home$\;
We perform a BFS search starting from $Home$
    \tcp{While encountering a node during exploration}
    \For{$v \in V$}{
        \If{$A_t(v) = 0$}
            {
            Place the lowest ID agent on $v$\;
            \tcp{$A_{t}(v)$ indicates the number of agents at $v$ at time $t$}
            Remaining agents continue exploration\;
            }
    }
\If{$C_{t'}(Home) = 0$}
{
\tcp{Can be figured out by the stationary agent at $Home$}
The initial agent (say $a_1$) at $Home$ moves forward, acting as a cleaner agent for the edges.
}
\end{algorithm2e}

\begin{theorem} \label{n result}
Let $G=(V, E)$ be the initial underlying static graph of the dynamic graph $\mathcal{G}$ on $\lvert V\rvert = n$ vertices, then $n$ initially co-located agents are sufficient to achieve monotone decontamination of $\mathcal{G}$.
\end{theorem}

\begin{proof}
    We assume that all agents begin co-located at a designated node, say $Home$, and that the graph $\mathcal{G}$ remains connected at every round (thus satisfying 1-interval connectivity). Using this assumption, we deploy $n-1$ agents using the \emph{Breadth First Search} strategy, which involves visiting all the 1-hop neighbors from $Home$, followed by all the 2-hop neighbors from $Home$, and so on, while keeping an agent stationary at each of these visited nodes.

Once each 1-hop neighbor of $Home$ gets an agent each, the node $Home$ gets \emph{fully decontaminated}. Next, according to our algorithm, the $Home$ agent, denoted as $a_1$ (considered the cleaner agent), starts visiting one node at a time in $G$ and fully decontaminates it by decontaminating each of its adjacent edges. While doing so, it encounters a contaminated disappeared edge from the current node, it waits for its reappearance and then decontaminates it. If multiple such disappeared edge exists, then it starts decontamination in increasing order of port numbers from the current node. Note that, once a node (say, $v\in V$) is fully decontaminated, it preserves this property. If $v=Home$, then all its adjacent nodes are already guarded by an agent each, and each adjacent edges are also decontaminated by those guarding agents during the BFS movement. So, further recontamination of either adjacent edges or $Home$ is not possible. If $v\neq Home$, then as well, all adjacent nodes of $v$ are already guarded by an agent each, so recontamination of either the adjacent edges or adjacent node, or $v$ is not possible. 

This shows that, $n$ initially co-located agents, executing the algorithm \textsc{Uni-Decontamination} is sufficient to fully decontaminate $\mathcal{G}$, in the \emph{finite edge reappearance} model.
\qed
\end{proof}


\noindent\textbf{Time Complexity:} The algorithm \textsc{Uni-Decontamination} takes time proportional to the size of the graph, plus the waiting delays caused by the finite time disappearing edges. The BFS strategy after incorporating the delays due to the disappearing edge requires $O(T\cdot (V+E))$ time, where $G=(V,E)$ is the underlying static graph, $T$ is the maximum time for which any edge can be disappeared.

 In the cleaning phase, the cleaner may need to wait at each vertex for up to $(\Delta-1)T$, where $\Delta$ is the maximum degree of $G$. For all vertices, this adds up to $O(V\cdot \Delta \cdot T)$. Thus, the total running time is $O(T\cdot (V+E)+V\cdot \Delta\cdot T)=O(n^2T)$, where $|V|=n$.

\begin{remark}
This trivial bound of $n$ agents is only meaningful as a universal upper bound.  In the remainder of this section, we focus on more refined results for the case $k \ll n$, where $k$ measures the number of dynamic threats (e.g., chords or feedback edges) in the underlying static topology.
\end{remark}

\begin{remark} 
Even if we are given the condition that the agents are unable to figure out whether the neighbouring nodes are contaminated or not, it doesn't matter. As the edge-reappearing time is finite, the agent on a node $v$ can simply wait until all the edges appear once and check them.
\end{remark}

\subsubsection{Case ($k<n$):} In this section, we discuss a monotone decontamination strategy for the case $k< n$, termed as \textsc{Modified-Decontamination}, where $\mu(G)=k$, $|V|=n$ and $G=(V,E)$ is the underlying static graph of the dynamic graph $\mathcal{G}$.

\noindent{\textbf{High-Level Idea of \textsc{Modified-Decontamination}:}} We now present an algorithm \textsc{Modified-Decontamination} that proposes to solve a monotonic decontamination of any arbitrary unknown graph $G$, with only $d+k
~(k\geq 1)$ agents where $k = \mu(G)$ is the cyclomatic number and $d$ is the diameter of the graph. So, if lets say $x$ (where suppose $x$ is sufficiently large) many agents are present, then we claim that among them only $d+k$ agents are sufficient to solve monotone decontamination, using \textsc{Modified-Decontamination} strategy. Unlike the previous strategy, this approach does not require that each vertex be permanently guarded. Instead, this method relies on a comparatively small team of agents working together to clean the entire graph while ensuring that no cleaned area ever gets recontaminated. The idea is: all agents begin at a certain root vertex (termed as $Home$), which is immediately decontaminated and guarded. The decontaminated area increases as agents traverse the network. At every step, a few agents stay behind to guard the boundary between the cleaned and contaminated parts of the graph. The set of these boundary nodes, are defined to be as the \emph{separator} nodes, and the agents guarding these nodes by remaining stationary are termed as \emph{separator} agents. At each separator vertex, the agent stationed there identifies the contaminated incident edges, including those that have temporarily disappeared but may reappear in the future. It may be noted that, a separator agent only remains stationary at a separator node, say $v$, if $C_t(v)>1$, i.e., when it finds the degree of adjacent contaminated edges (including the disappeared contaminated edges) is more than 1. Otherwise, for $C_t(v)=1$, it continues to sweep through, i.e., in turn decontaminating the separator node. These separator agents prevent contamination from spreading back into the clean zone, even if edges reappear due to the dynamic nature of the graph. For the case, when $C_t(v)>1$, the separator agent waits till $t'>t$, when $C_{t'}(v)=1$.

Meanwhile, the free agents (i.e., the ones that are not guarding any separator node as a separator agent) sweep through the graph carefully exploring and cleaning new vertices and edges, until they encounter a separator vertex, $u$ say, with $C_{t''}(u)>1$. The main task of all the agents, starting from $Home$ is to reduce the separator nodes in the graph, until it becomes null. In order to move along the graph, they perform \textit{depth first search} (DFS) traversal. During this traversal, whenever an agent encounters a node, which is fully decontaminated (i.e., the node and all its neighbors are decontaminated), then it backtracks.

 In the following part, we show that \textsc{Modified-Decontamination} monotonically decontaminates the dynamic graph $\mathcal{G}$. 

\begin{algorithm2e}[h]\small
\caption{\textsc{Modified-Decontamination}($k,Home$)}
\label{algo:mod_decon}
\KwIn{$Home$ node, $d+k$ agents}

Let $\Sigma$ denote the separator vertices in $\mathcal{G}$\;
Set $\Sigma=V$\;
Initially $d+k$ agents are at $Home$\;
    \For{$v \in \Sigma$}{
    \If{$C_t(v)>1$}
    {
    Choose the lowest port among them, say that leads to $u\in N(v)$\;
        \tcp{Underlying traversal strategy is a DFS algorithm}
        
        Lowest Id agent remains, till $t'$ when $C_{t'}(v)=1$\;
        
        Remaining agents move through the edge $(v,u)$\;
        \tcp{Since number of contaminated (unexplored) edges decrease, this process is finite for a fixed $v$}
        }
        \uElseIf{$C_t(v)=1$}
        {
        The stationary agent at $v$, moves along $(v,u)$\;
        $\Sigma=\Sigma - \{v\}$\;
        }
        \Else
        {
        All agents from current node, backtracks from this node\;
        }
    }
\end{algorithm2e}


\begin{theorem}\label{d+k result}
Let $ \mathcal{G} \subseteq G $ be a dynamic connected graph of the static graph $ G $, with cyclomatic number $ \mu(G) = k $ and diameter $d$. Then, $ d+k $ agents are sufficient to perform monotone decontamination starting from any initial node.
\end{theorem}

\begin{proof}
Fix a spanning tree $ T=(V,E') \subseteq G $ of the static graph $ G $. Let the set of \emph{feedback} edges be $ \Tilde{E} = E \setminus E' $, with $ |\Tilde{E}| = \mu(G) = k $. These feedback edges are the only ones that can introduce cycles and potential recontamination, due to the dynamicity.

Our algorithm \textsc{Modified-Decontamination} works by placing all $ d+k $ agents at the initial vertex $ Home $. At this point, the collection of all separator vertex is $\Sigma=V$. Next, we will show that following invariant holds:

\begin{enumerate}
    \item  \textbf{Inv-1:} Whenever a node, say $v\in V$ is visited for the first time by an agent at round $t$, our algorithm ensures that within the time interval $[t,t']$ (where $t<t'$), the node $v$ is \emph{fully decontaminated}.

    \item \textbf{Inv-2:} Every node in $G$ is visited by at least one agent, during the execution of our algorithm.
\end{enumerate}

Proof of these invariants follows from Lemma \ref{lem:full_decon} and \ref{lem:full_visit}.

So, \textbf{Inv-1} and \textbf{Inv-2} will only mean that the graph is fully decontaminated if monotonicity is preserved. 

We claim that our algorithm preserves monotonicity. Let $v$ be a decontaminated node, then there is no possibility of $v$'s recontamination, as it is not left unguarded by an agent until $v$ gets fully decontaminated. This also guarantees that, there is no possibility of recontamination in a fully decontaminated node. It is because, any path of recontamination must pass through at least one decontaminated node. Lastly, by the same argument, we can conclude that any decontaminated edge does not gets recontaminated. Hence, monotonicity holds in our algorithm. \qed\end{proof}


\begin{lemma}\label{lem:full_decon}
Whenever a node, say $v\in V$ is visited for the first time by an agent at round $t$, our algorithm ensures that within the time interval $[t,t']$ (where $t<t'$), the node $v$ is \emph{fully decontaminated}.
\end{lemma}

\begin{proof}
In order to prove this lemma, based on the number of adjacent contaminated edges with respect to $v$ at time $t$, we have the following cases:

\textbf{Case-1 $[C_t(v)>1]$ :} Let $v$ be $j$-th node in the diameter of $G$. We consider two instances, first when $k=0$ and second when $k\ge 1$. 

If $k=0$, then in the worst case $j-1$ agents are already guarding the $j-1$ nodes in the diameter path from $Home$.
Proceeding further in the trail, we (in worst case) encounter $d-j$ many nodes where we keep 1 agent on each node. Since the agent can't proceed further, it backtracks to a node $u$, with $C_{t_1}(u)$ where $t_1$ ($t<t_1<t'$) is the time $u$ was first visited by an agent. Following the algorithm, one agent $a_u$ was placed at $u$. The \emph{backtracking} agent reaches $u$ at time $t'_1$. Now, $C_{t'_1}(u)\leq C_t(u)-1$. Recursively, there exists a time $t''_1$ ($\ge t'_1$) when $C_{t''_1}(u)=0$. This holds for every node in $G$, hence it is sufficient with $d$ agents to fully decontaminate $G$. 

If $k\ge 1$, by similar idea suppose in the worst case $j+i-1$ agents are already guarding the $j-1$ nodes in the diameter path from $Home$, and $i$ many agents are guarding $k-i$ feedback edges. In this situation, proceeding along the trail, in the worst case, the group of agents may encounter $d-j$ diameter nodes where an agent must be kept, and $k-i$ feedback edges, where remaining $k-i$ agents are kept. Since, the agent last kept on the farthest trial from $Home$ cannot move further, it backtracks, to a node $u$. By earlier argument, $C_{t''_1}(u)$ becomes 0. Hence, recursively we can argue that $G$ will be fully decontaminated with $d+k$ agents.

\textbf{Case-2 $[C_t(v)=1]$ :} The agent reaching $v$ at time $t$, moves along the decontaminated node, making $C_{t+1}(v)=0$, hence fully decontaminating $v$.
\par

So, inevitably we show that there exists a time $t'>t$, when $C_{t'}(v)$ becomes 0. Hence, this proves the theorem.\qed\end{proof}

\begin{lemma}\label{lem:full_visit}
    Every node in $G$ is visited by at least one agent, during the execution of our algorithm.
\end{lemma}

\begin{proof}
    Let $v$ be a node, which is never visited by any agent, while executing the algorithm \textsc{Modified-Decontamination}. Let $u$ be an adjacent node of $v$, which is first visited by any agent at round $t$. Now, by the argument $v$ is never visited implies that $(u,v)$ is contaminated. So, this means $C_t(v)\ge 1$. By Lemma \ref{lem:full_decon}, there exists a round $t'>t$, at which $v$ gets fully decontaminated. Hence, this contradicts that the node $v$ is never visited. \qed\end{proof}

\begin{corollary}[Time Complexity]
\textsc{Modified-Decontamination} takes time proportional to the size of the graph, plus the waiting delays caused by the finite time disappearing edges. Incorporating those delays into DFS yields $O(T\cdot (V+E))$ time, where $G=(V,E)$ is the underlying static graph, $T$ is the maximum time for which any edge can disappear.
\end{corollary}

In the following example, we show a graph $G$ in which $d+k-1$ agents using algorithm \textsc{Modified-Decontamination} cannot decontaminate $G$.

\begin{example}
We see, for any \emph{monotone} decontamination procedure starting at root (red squared), we need to keep 1 agent at each of the $d-1$ nodes, and need at least $k+1$ agents to decontaminate the $k$-cycle connected to the basepoint, with a strategy being sending an agent to each of the $k$ nodes and using 1 as an explorer. Therefore, total number of agents required for \emph{monotone decontamination} of graph in Figure \ref{fig:ub_tree_example} is $(d-1) + k + 1 = d+k$.

\begin{figure}[h!]
  \centering
\begin{tikzpicture}[scale=0.75,
    every node/.style={circle, draw, thick, inner sep=2pt, fill=white}
]

\node[circle, draw, thick] (a) at (0,0) {};

\node[draw=red, rectangle, fit=(a), inner sep=4pt, fill opacity=0, label=left:{$1$}] {};

\node (aL1) at (0.5,-0.5) {};
\node (aL2) at (0, -1.5) {};
\node (aR1) at (0.5,0.5) {};
\node (aR2) at (0,1.5) {};
\draw[thick] (a)--(aL1);
\draw[thick, dashed] (aL1)--(aL2);
\draw[thick] (a)--(aR1);
\draw[thick, dashed] (aR1)--(aR2);

\node[label=above left:{$1$}] (b) at (1,0) {};
\node[label=above left:{$1$}] (c) at (2,0) {};
\node[label=above left:{$1$}] (d) at (4,0) {};

\draw[thick] (a)--(b);
\draw[thick] (b)--(c);
\draw[thick, dashed] (c)--(d);

\node (bL1) at (1.5,1) {};
\node (bL2) at (1, 1.5) {};
\node (bR1) at (1.5,-1) {};
\node (bR2) at (1, -1.5) {};
\draw[thick] (b)--(bL1);
\draw[thick, dashed] (bL1)--(bL2);
\draw[thick] (b)--(bR1);
\draw[thick, dashed] (bR1)--(bR2);

\node[label={[label distance = 5pt]above left:{$1$}}] (t) at (5,0) {};
\node (v1) at (6.75,1.5) {};
\node (v2) at (6,0) {};
\node (v3) at (6.75,-1.75) {};

\node (tr) at (4.5,-1.25) {};
\draw[thick, dashed] (t)--(tr);
\node (tl) at (4.5,1.25) {};
\draw[thick, dashed] (t)--(tl);

\draw[thick] (d)--(t);
\draw[dotted, line width=1pt] (2,1) -- (4,1);
\draw[dotted, line width=1pt] (2,-1) -- (4,-1);

\coordinate (a_shift) at (0,-2);
\coordinate (t_shift) at (5.2,-2);
\draw[decorate,decoration={brace,mirror,amplitude=8pt}] 
  (a_shift.south east) -- (t_shift.north east) node[draw=none, midway,yshift=-0.8cm] {$d-1$};

\draw[thick] (t)--(v1);
\draw[->,thick,shorten <=5pt,shorten >=5pt] 
  (t) to[bend left=20] node[draw=none, midway, sloped, above=3pt] {1} (v1);
\draw[thick] (t)--(v2);
\draw[->,thick,shorten <=5pt,shorten >=5pt] 
  (t) to[bend right=20] node[draw=none, midway, sloped, right=20pt] {1} (v2);
\draw[thick] (t)--(v3);
\draw[->,thick,shorten <=5pt,shorten >=5pt] 
  (t) to[bend right=20] node[draw=none, midway, sloped, below=3pt] {1} (v3);

\draw[thick] (v1)--(v2)--(v3)--(v1);

\draw[red, thick, decorate, decoration={snake, amplitude=2, segment length=6pt}] (v1)--(v2);
\draw[red, thick, decorate, decoration={snake, amplitude=2, segment length=6pt}] (v2)--(v3);
\draw[red, thick, decorate, decoration={snake, amplitude=2, segment length=6pt}] (v3)--(v1);

\coordinate (v3_shift) at (6.75,-2);
\draw[decorate,decoration={brace,mirror,amplitude=8pt}] 
  (t_shift.south west) -- (v3_shift.south east) node[draw=none, midway,yshift=-0.7cm] {$C_k$};
\end{tikzpicture}
\caption{The graph here represents a case of spanning tree where we need the exact \emph{upper bound} number of agents, i.e., $d+3$, where $d,\ k=3$. It consists of a path of $d-1$ nodes, branches which contain paths of $\leq d$ nodes; attached to $C_k$ graph where $k=3$.}
\label{fig:ub_tree_example}

\end{figure}
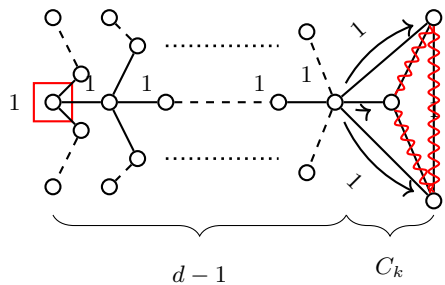

\end{example}


\subsection{Indefinite Edge Disappearance Model}
\label{subsec:Indefinite-Edge-Disappearance-Model}

In this section we propose a monotone decontamination algorithm, that requires $d+2k$ agents, initially co-located at a node of the underlying dynamic graph $\mathcal{G}$. We define $G=(V,E)$ to be the initial static graph of $\mathcal{G}$, with diameter $d$ and $\mu(G)=k$. We term the algorithm as \textsc{Infinite-Decontamination}.

\noindent\textbf{High Level Idea of \textsc{Infinite-Decontamination}}: We present the algorithm that proposes to decontaminate an underlying unknown dynamic graph $\mathcal{G}$ with $d+2k$ agents, where $k$ is the cyclomatic number of the initial static graph $G$ of $\mathcal{G}$, and $d$ is the diameter of $G$. In other words, if $x$ (where $x$ is sufficiently large) many agents are present at the initial node, then we claim that among them $d+2k$ are sufficient to solve monotone decontamination using the strategy \textsc{Infinite-Decontamination}.

Here, similar to the previous algorithm \textsc{Modified-Decontamination} we rely on a DFS strategy to decontaminate the graph. The exact strategy follows, but unlike the earlier strategy, there are certain changes. Let us suppose $v$ be the node where a set of $x$ agents arrive at time $t$. Based on this we can have the following cases: $C_t(v)\ge 1$ or $C_t(v)=0$.

If $C_t(v)=0$ and there are $M_t(v)$ many missing edges adjacent to $v$, where $\{M_t(v)\}$ indicates the collection of ports corresponding to these missing edges at $v$ at time $t$ and $M_t(v)$ indicates its cardinality. Then at most $M_t(v)$ agents remain at $v$, guarding these particular missing edges until they reappear. Remaining agents at $v$ find an alternate route to move along, since at any round $\mathcal{G}$ must be connected.

If $C_t(v)\ge \alpha$ (where $\alpha \ge 1$) and $M_t(v)$ many missing edges are adjacent to $v$. Then atmost $M_t(v)$ agents remain at $v$. Among these $M_t(v)$ edges, some may be contaminated and some may be decontaminated. The agents waiting for decontaminated missing edges wait until their respective missing edges appear. On the other hand, the agents waiting for contaminated missing edges not only wait until their respective edges appear but also until the time, say $t'$, when there exists one agent at $v$, that remains stationary for $v$ to be fully decontaminated.

\begin{algorithm2e}[H]\small 
\caption{\textsc{Infinite-Decontamination}($k,Home$)
\label{algo:inf_decon}}
\KwIn{$Home$ node, $d+2k$ agents}

Initialize cleaned set $I = \{Home\}$\;
Place all $d+2k$ agents at $Home$ \;
Set $\Sigma=V$ \;
Set $\{C_0(v)\}\leftarrow$ all adjacent edges of $v$, and $C_0(v)=\delta_v$, for all $v$ \; 
\tcp{$\{C_t(v)\}$ indicates the set of all ports corresponding to which edges are contaminated}

Set $\{M_0(v)\}=NULL$\; \tcp{ $\{M_t(v)\}$ is the set of all ports corresponding to which the edges are missing at time $t$}

\For{$v \in \Sigma$}
{
    \tcp{Underlying traversal strategy is a DFS algorithm}
    \uIf{$C_t(v) \ge 1$}
    {
    Determine the set of $\{M_t(v)\}$, i.e., missing adjacent edges and keep at most that many agents on $v$ \;
    \tcp{These agents are for guarding the missing edges, they will advance when these reappear}
    \tcp{Checking number of visible contaminated edges,}
        \uIf{$|\{C_t(v)\} \cap (\{\delta_v\} \setminus \{M_v(t)\})| > 1$}
        {
        Place one more agent (the lowest ID) at $v$ \;
        Advance with the remaining agents to one of the \emph{visible contaminated} ports.
        \tcp{The agent left was for guarding the other edges from recontamination.}
        }
        \uElseIf{$|\{C_t(v)\} \cap (\{\delta_v\} \setminus \{M_v(t)\})| = 1$}
        {
        All remaining agents advance towards this one port\;
        }
        \Else
        {
        Since we have no other edges to explore, remaining agents search for a new route for further exploration\; \tcp{can be found as $\mathcal G_t$, i.e., dynamic graph of $\mathcal G$ at time $t$, is connected}
        }
    }
    \Else
    {
    We place atmost $M_t(v)$ agents on $v$ for the missing edges and proceed with the remaining agents to search for further exploration, either by backtracking or moving forward \;
    }
    Whenever the missing edges re-appear (at time $t'>t$), our agents situated at the vertex adjacent to them move through their corresponding edges to decontaminate or to explore them iff $A_{t'}(v) = M_{t'}(v)$ \;
    \tcp{$A_{t'}(v)$ indicates the number of agents at $v$ at time $t'$}
}
\end{algorithm2e}

Only if both these conditions hold, then they are allowed to move and decontaminate, after their respective edges has reappeared. 

Now, among the remaining $\delta_v-M_t(v)$ many adjacent edges, there are further three sub-cases: \textbf{SubCase-1}: none edge is contaminated, \textbf{SubCase-2}: a single edge is contaminated, \textbf{SubCase-3}: more than one edges are contaminated. 

\textbf{SubCase-1}: If no edge is contaminated, then all agents, except $M_t(v)$ many agents, find a route from $v$, since $\mathcal{G}$ is connected. 

\textbf{SubCase-2}: If a single edge is contaminated, then all the agents, except $M_t(v)$, move along this contaminated edge. 

\textbf{SubCase-3}: If more than one edges are contaminated, then the lowest Id agent (say $a_1$), among all the agents at $v$, except the $M_t(v)$ many waiting agents, remain stationary, and the remaining agents move along the lowest port labeled contaminated non-missing edge. The agent $a_1$ remains stationary at $v$ until either \textbf{SubCase-1} or \textbf{SubCase-2} occurs.


\begin{theorem}\label{d+2k result}
    $\mathcal{G}$ be a dynamic graph, with underlying initial static graph $G=(V,E)$ with cyclomatic number $\mu(G)=k$ and diameter $d$. Our algorithm \textsc{Infinite-Decontamination} monotonically decontaminates $\mathcal{G}$ with $d+2k$ agents.
\end{theorem}

\begin{proof}
    Fix a spanning tree $T=(V,E')\subseteq G$ of the static graph $G$. The set of \emph{feedback} edges are $\Tilde{E}=E\setminus E'$, with $|\Tilde{E}|=\mu(G)=k$. These feedback edges are the only ones for which cycles may be introduced, or potential recontamination may occur. 

    As per our algorithm \textsc{Initial-Decontamination}, all $d+2k$ agents are initially placed at $Home$. At this point the collection of $\Sigma=V$. Next, we have the following invariants:

    \begin{enumerate}
        \item \textbf{Inv-1}: Whenever a node, say $v\in V$ is visited for the first time by an agent at time $t$, our algorithm ensures that from time $t$ onwards, at least one agent remains stationary at $v$ until $v$ is \emph{fully decontaminated}.

        \item \textbf{Inv-2}: Every node in $G$ is visited by at least one agent, during the execution of \textsc{Infinite-Decontamination}.
    \end{enumerate}

So, \textbf{Inv-1} (refer Lemma \ref{lem:inf-station_decon}) and \textbf{Inv-2} (refer Lemma \ref{lem:inf-each-visited}) implies that, all existing edges in $\mathcal{G}$ will be fully decontaminated.

A node $v$ which is decontaminated, will remain decontaminated, as there always exists at least one agent guarding $v$, until $v$ is fully decontaminated. This nullifies the possibility of recontamination, as any path of recontamination must pass through at least one decontaminated node. Hence, monotonicity of our algorithm holds.
\qed
\end{proof}

\begin{lemma}\label{lem:inf-station_decon}
    Whenever a node, say $v\in V$ is visited for the first time by an agent at time $t$, our algorithm ensures that from time $t$ onwards, at least one agent remains stationary at $v$ until $v$ is \emph{fully decontaminated}.
\end{lemma}

\begin{proof}
    To prove this lemma, we consider a vertex $v$, at which whenever an agent, say $a_1$, visits it, finds $C_t(v)>1$ and there exists $\alpha$ many missing edges adjacent to $v$. These missing edges can either be contaminated or decontaminated. In either case, as per our algorithm, whenever an agent finds a missing edge (may be contaminated or decontaminated) with no agent waiting for it, then the agent waits for it. So, since at any time at most $k$ such missing edge can exist so, and among them $\alpha$ are adjacent to $v$. So, inevitably for remaining $k-\alpha$ missing edges, $2(k-\alpha)$ agents will act as a guard. Moroever, the graph $\mathcal{G}$ being connected, with the agents employing a strategy similar to DFS, guarantees that $\alpha$ many agents inevitably reach $v$ for those $\alpha$ missing edges. Again, since the remaining graph is acyclic (as all cycles can be formed using these feedback edges), the remaining $d$ agents are sufficient to decontaminate it. Hence, this proves that, until the node $v$ is fully decontaminated, at least one agent remains as a guard, waiting for its full decontamination.\qed\end{proof}

\begin{lemma}\label{lem:inf-each-visited}
    Every node in $G$ is visited by at least one agent, during the execution of \textsc{Infinite-Decontamination}.
\end{lemma}

\begin{proof}
    Let us consider a node $v \in V$. Say $\exists\ \alpha$ many paths from $Home$ to $v$. Our Algorithm \ref{algo:inf_decon} (underlying DFS) ensures that agents start proceeding towards $v$ through one of these paths. When blocked, we place (sufficient, from Lemma \ref{lem:inf-station_decon}) agents which move forward sometime after the edges reappear. Meanwhile, rest of the agents try to proceed through the other paths and follow the same rule there. 

    Since forward motion of the agents are ensured, we can guarantee that the agents reach $v$ in finite time starting from $Home$.\qed\end{proof}

\begin{remark}
Our algorithm guarantees that all the vertices are decontaminated, but it did not guarantee all the edges to be decontaminated, because the adversary may not reappear some edges, after the agents guard those adjacent nodes. So, it is impossible to decontaminate those edges.    
\end{remark}

\begin{lemma}\label{lemma: InfTime}
Algorithm \textsc{Infinite-Decontamination} ensures that each node in the underlying dynamic graph $\mathcal{G}$ is decontaminated within $O(n^2)$ rounds.
\end{lemma}

\begin{proof}
    Consider a cycle $C$ of length $l$ inside $G$ (where $G$ is the initial static graph of $\mathcal{G}$). Now, whenever the agents visit a node in $C$, they either perform clockwise or counter-clockwise movement to visit the remaining nodes of $C$. Now, after each movement, the adversary can disrupt the agents from visiting further by disappearing an edge. So, in this phenomenon, our algorithm instructs an agent among these group of agents to remain stationary, whereas the other agents move in the opposite direction to find an alternate route. Again, the adversary may disappear an edge along the opposite direction movement of the agents, this disappearance inturn reappears the earlier edge (to maintain connectivity in $C$) which helps the earlier stationary agent to move forward, and visit a new node.
    But they might not be able to move further forward because the new vertex might have \emph{contamination degree} more than the previous one and hence requires more agents to explore. This to and fro movement may continue. So, in order to visit each node in $C$, it takes $O(l^2)$ time for the agents, executing \textsc{Infinite-Decontamination}. Now, since $l\le n$, we can conclude that in the worst case to visit every node in $\mathcal{G}$, it takes $O(n^2)$ time.
\qed\end{proof}


\section{Conclusion}

In this paper, we study the monotone decontamination on two models of dynamicity \textit{finite time edge appearance} (i.e., FTEA) and \textit{indefinite edge disappearance} (i.e., IDED) models. The FTEA model is a special case of IDED, accordingly the bounds obtained by us reflects the added complexities for the IDED model. In all our algorithms as well as lower bounds, we assume the underlying graphs to be arbitrary and the agents have no initial knowledge about it. Accordingly, the bounds obtained in this paper are among the first results to be obtained for network decontamination in dynamic networks. There are many directions of future work that can be explored. Among them, the most important fact is that, our results are not tight, so an important future work must be to obtain tight bounds. Another possible direction is to explore different other models of dynamicity, in each of these models this problem can be investigated.


\bibliographystyle{splncs04}
\bibliography{mybib}















\end{document}